\documentclass{llncs}

\usepackage{
amsmath,
amscd,
amssymb,
}
\usepackage{comment}
\usepackage{tikz}
\usetikzlibrary{positioning,arrows,calc}
\usepackage{xspace}

\usepackage{graphicx}

\usepackage{epstopdf}

\usepackage{url}

\newcommand{\Z}{\mathbb{Z}}
\newcommand{\N}{\mathbb{N}}

\newcommand{\ID}{\mathrm{id}}

\newcommand{\clone}[1]{\left\lceil#1\right\rceil}
\newcommand{\gen}[1]{\left\langle#1\right\rangle}
\newcommand{\id}{id}
\newcommand{\controlled}[2]{f_{#1,#2}}
\newcommand{\cp}[2]{CP(#1,#2)}
\newcommand{\Cone}{B(A)}
\newcommand{\Ctwo}{Cons(A)}
\newcommand{\Cthree}{Even(A)}
\newcommand{\Cfour}{ECons(A)}

\newcommand{\Sym}{\mathrm{Sym}}
\newcommand{\Alt}{\mathrm{Alt}}
\newcommand{\sym}{\mathrm{Sym}}
\newcommand{\alt}{\mathrm{Alt}}
\newcommand{\what}{revital}

\newcommand{\UTU}{
 Department of Mathematics and Statistics, University of Turku, Finland.
 }

\newcommand{\LINZ}{
 Institute for Algebra, Johannes Kepler University Linz, Austria
and Time's Up Research, Linz, Austria.
 }

\newcommand{\CHILE}{
Center for Mathematical Modeling, University of Chile, Santiago, Chile.
}

\begin{document}

 \mainmatter

 \title{Strongly Universal Reversible Gate Sets \thanks{The authors would like to acknowledge the contribution of the COST Action IC1405
 This work was partially funded by
 Austrian national research agency FWF research grants P24077 and P24285, and by FONDECYT research grant 3150552.}}
 \titlerunning{Strongly Universal Gates}

 \author{Tim Boykett \inst{1}
    \and Jarkko Kari \inst{2} 
    \and Ville Salo \inst{2,3} 
 }

\institute{\LINZ \and \UTU \and \CHILE}

\maketitle

\begin{abstract}
It is well-known that the Toffoli gate and the negation gate together yield a universal gate set, in the sense that every permutation
of $\{0,1\}^n$ can be implemented as a composition of these gates. Since every bit operation that does not use all of the bits performs
an even permutation, we need to use at least one auxiliary bit to perform every permutation, and it is known that one bit is indeed enough.
Without auxiliary bits, all even permutations can be implemented. We generalize these results to non-binary logic: If $A$ is a finite set
of odd cardinality then a finite gate set can generate all permutations of  $A^n$ for all $n$, without any auxiliary symbols.
If the cardinality of $A$ is even then, by the same argument as above, only even permutations of $A^n$ can be implemented for large $n$,
and we show that indeed all even permutations can be obtained from a finite universal gate set. We also consider the conservative case,
that is, those permutations of $A^n$ that preserve the weight of the input word. The weight is the vector that records how many times
each symbol occurs in the word. It turns out that no finite conservative gate set can, for all $n$,
implement all conservative even permutations
of  $A^n$ without auxiliary bits. But we provide a finite gate set that can implement all those conservative permutations
that are even within each weight class of $A^n$.
\end{abstract}

\section{Introduction}

The study of reversible and conservative binary gates was pioneered in the 1970s and 1980s by
Toffoli and Fredkin \cite{FrTo82,toff80}.
Recently, Aaronson, Greier and Schaeffer \cite{aaronsonetal15} described all binary gate sets closed under
the use of auxilliary bits, as a prelude to their eventual goal of classifying these gate sets in the quantum case. It has been noted that ternary gates have similar, yet distinct properties \cite{yangetal05}.

In this article, we consider the problem of finitely-generatedness of various families of reversible logic gates without using auxiliary bits.
In the case of a binary alphabet, it is known that the whole set of gates is not finitely generated, but the family of gates that perform an even permutation of $\{0,1\}^n$ is \cite{aaronsonetal15,xu15}. In \cite{yangetal05}, it is shown that for the ternary alphabet, the whole set of reversible gates is finitely generated. In this paper, we look at gate sets with arbitrary finite 
alphabets, and prove the natural generalization: the whole set of gates is finitely generated if and only if the alphabet is odd, and in the case of an even alphabet, the even permutations are finitely generated.

In \cite{xu15}, it is proved that in the binary case the conservative gates, gates that preserve the numbers of symbols in the input (that is, its weight), are not finitely generated, even with the use of `borrowed bits', bits that may have any initial value but must return to their original value in the end. On the other hand, it is shown that with bits whose initial value is known (and suitably chosen), all permutations can be performed. We prove for all alphabets that the gates that perform an even permutation in every weight class are finitely generated, but the whole class of permutations is far from being finitely generated (which implies in particular the result of \cite{xu15}).

Our methods are rather general, and the proofs both in the conservative case and the general case follow the same structure. The negative aspect of these methods is that our universal gates are not the usual ones, and for example in the conservative case, one needs a bit of work (or computer time) to construct our universal gate family from the Fredkin gate.

We start by introducing our terminology, taking advantage of the
concepts of clone theory \cite{szendrei} applied to bijections
as developed in \cite{boykett15}, leading to what we call \emph{reversible clones} or \emph{revclones}, and \emph{reversible iterative algebras} or
\emph{\what{}s}.
We generalize the idea of the Toffoli gate and Fredkin gate to what we call `controlled permutations' and prove a general induction lemma showing that if we can a single new control wire to a controlled permutation, we can add any amount. 
We then show two combinatorial results about permutation groups that allow us
to simplify arguments about \what{}s.
This allows us to describe generating sets for various revclones and \what{}s of interest,
with the indication that these results will be useful for more general \what{} analysis,
as undertaken for instance in \cite{aaronsonetal15}.
While theoretical considerations show that finite generating sets do not
exist in some cases, in other cases explicit computational searches are able to provide
small generating sets.

\section{Background}

Let $A$ be a finite set. We write $S_A$ or $\Sym(A)$ for the group of permutations or bijections of
$A$, $S_n$ for $Sym(\{1,\dots,n\})$ and
$\Alt(A)$ for the group of even permutations of $A$, $A_n=Alt(\{1,\dots,n\})$. 
We will compose functions from left to right.
Let $B_n(A) = \{f:A^n\rightarrow A^n \;\vert\; f \mbox{ a bijection}\}=\Sym(A^n)$
be the group of $n$-ary bijections on $A^n$, and let
$B(A) = \cup_{n\in \N} B_n(A)$ be the collection of all bijections on powers of $A$.
We will call them \emph{gates}. We denote by
$\gen{X}$ the group generated by $X\subseteq B_n(A)$, a subgroup of $B_n(A)$.

Each $\alpha \in S_n$ defines a \emph{wire permutation} $\pi_\alpha \in B_n(A)$
that permutes the coordinates of its input according to $\alpha$:
$$
\pi_\alpha(x_1,\dots,x_n) = (x_{\alpha^{-1}(1)},\dots, x_{\alpha^{-1}(n)}).
$$
The wire permutation $\id_n=\pi_{()}$ corresponding to the identity permutation $()\in S_n$ is the $n$-ary identity map.
Conjugating $f\in B_n(A)$ with a wire permutation $\pi_\alpha\in B_n(A)$ gives
$\pi_\alpha \circ f\circ \pi_\alpha^{-1}$, which we call a
\emph{rewiring} of $f$.
Rewirings of $f$ correspond to applying $f$ on arbitrarily ordered input wires.

Any $f\in B_\ell(A)$ can be applied on $A^n$ for $n>\ell$ by applying it on selected $\ell$ coordinates while leaving
the other $n-\ell$ coordinates unchanged. Using the
clone theory derived  terminology in \cite{boykett15} we first define, for any $f\in B_n(A)$ and $g\in B_m(A)$, the
parallel application $f \oplus g \in B_{n+m}(A)$ by
\begin{align*}
(f \oplus g)(x_1,\dots x_{n+m}) = (&f_1(x_1,\dots,x_n), \dots, f_n(x_1,\dots,x_n), \\
& g_1(x_{n+1},\dots,x_{n+m}), \dots,g_m(x_{n+1},\dots,x_{n+m})).
\end{align*}
Then the \emph{extensions} of $f\in B_\ell(A)$ on $A^n$ are the rewirings of $f\oplus \id_{n-\ell}$.

Let $P\subseteq B(A)$. We denote by $\clone{P}\subseteq B(A)$ the set of gates that can be obtained
from the identity $\id_1$ and the elements of $P$ by compositions of gates of equal arity and by
extensions of gates of arities $\ell$ on $A^n$, for  $n\geq \ell$. Clearly $P\mapsto \clone{P}$
is a closure operator. Sets $P\subseteq B(A)$ such that $P=\clone{P}$ are called \what{}s.
We say that $P$ \emph{generates} \what{} $C$ if $C=\clone{P}$. We say that  \what{} $C$ is \emph{finitely generated}
if there exists a finite set $P$ that generates it.

To relate the concepts to clone theory, one defines the generalized compositions of permutations of arbitrary arities
as follows: Let $f\in B_n(A)$ and $g\in B_m(A)$. For $k\leq \min(m,n)$, let
$f \circ_k g \in B_{n+m-k}(A)$ be defined by
\begin{align*}
\nonumber
 f \circ_k g = (&g_1(f_1(x_1,\dots,x_n),\dots,f_k(x_1,\dots,x_n),x_{n+1},\dots,x_{n+m-k}),\dots, \\
    &g_m(f_1(x_1,\dots,x_n),\dots,f_k(x_1,\dots,x_n),x_{n+1},\dots,x_{n+m-k}), \nonumber\\
      &f_{k+1}(x_1,\dots,x_n),\dots,f_n(x_1,\dots,x_n))
\end{align*}
If $n=m=k$ this is the usual composition $ f \circ g$.
We call $(B(A); \{\oplus, \circ,\pi_\alpha \;\vert\; \exists n \in \N: \alpha \in S_n \})$ the \emph{full reversible clone on $A$}
and any subalgebra a reversible clone on $A$, or simply a \emph{revclone}.\footnote{In this paper, we are more concerned with the set of functions in a revital or revclone, rather than the particular signatures chosen, and thus have chosen this revclone signature due to its (apparent) simplicity -- in clone theory, finite signatures are preferred, see \cite{boykett15} for such a revclone signature.}
Every revclone is a \what{} and, in fact, revclones are precisely the \what{}s that contain all wire permutations
$\pi_\alpha$ or, equivalently, the \what{}s that contain the wire permutation $\pi_{(1\;
2)}\in B_2(A)$ that swaps two wires. Note that $\clone{\pi_{(1\; 2)}}$ is exactly the set of wire
permutations.
It follows that if $P$ generates $C$ as a revclone, then $P'=P\cup\{\pi_{(1\; 2)}\}$ generates it as a
\what{}, so there is no difference in the finitely-generatedness of
a revclone when we consider it as a \what{} instead of a revclone.

 We sometimes refer to general elements of $B_n(A)$ as \emph{word permutations} to distinguish them from the wire permutations.
 In particular, by a wire swap we refer to a function $f : A^2 \to A^2$ with $f(a,b) = (b,a)$ for all $a, b \in A$ (or an extension of such a function),
 while a word swap refers to a permutation $(u \; v) \in B_n(A)$ that swaps two individual words of the same length.
 Of course, a wire swap is a composition of word swaps, but the converse is not true. Similarly, and more generally,
 we talk about \emph{wire and word rotations}. A \emph{symbol permutation} is a permutation of $A$.

We are interested in finding out if some naturally arising \what{}s are finitely generated. First of all, we have the \emph{full \what{}}
$B(A)$ and the \emph {alternating \what{}} $\Cthree = \bigcup_n \Alt(A^n)$ that contains all even
permutations.

We also consider permutations that conserve the letters in their inputs.
For any $n\in \N$, define $w_n:A^n\rightarrow \N^A$, such that for all $x\in A^n$, $a\in A$,
$w_n(x)(a)$ the number of occurences of $a$ in $x$. We say $w_n(u)$ is the \emph{weight} of the word $u$.
A mapping $f\in B_n(A)$ is \emph{conservative} if for all $x\in A^n$, $w_n(f(x)) = w_n(x)$, we let $Cons_n(A)\subseteq B_n(A)$ be the set of conservative maps of arity $n$.
Then $\Ctwo = \cup_{n\in \N}Cons_n(A)$ is the \emph{conservative \what}.
We also consider the set of conservative permutations
that perform an even permutation on each weight class, denoted by $\Cfour$, called the \emph{alternating conservative revital}.

A wire swap $\alpha$, on $A^n$, has parity $\frac{\vert A\vert (\vert A\vert-1)}{2}|A|^{n-2}$. When $n = 2$, this is even only when $|A| \equiv 0$ or $|A| \equiv 1 \pmod 4$. It follows that $\Cthree$ is a revclone only when $|A| \equiv 0$ or $|A| \equiv 1 \pmod 4$.
The \what{} $\Cfour$ is never a revclone because swaps are odd permutations on the words with a single symbol different from the others.

Furthermore, for any $k\in \N$, we can define the mappings that are \emph{conservative modulo $k$}
by replacing $\N$ with $\Z_k$ in the above definition.
We will write $Mod_k(A)$ for these maps.



Using the terminology in~\cite{xu15}, we say that gate
$f\oplus\id_k\in B_{n+k}(A)$ computes $f\in B_n(A)$ using $k$
\emph{borrowed} bits. The borrowed bits are auxiliary symbols in the computation of $f$
that can have arbitrary initial values, and at the end these values
must be restored unaltered. Regardless of the initial values of the borrowed bits, the permutation $f$ is computed on the
other $n$ inputs. We have cases where borrowed bits help
(Corollary~\ref{cor:borrowedcor}) and cases where they don't (Theorem~\ref{thm:XuGeneralization}).

A \emph{hypergraph} is a set $V$ of vertices and a set $E$ of edges, $E \subseteq {\cal P}(V)$.
A $k$-hypergraph is a hypergraph where every edge has the same size, $k$.
A 2-hypergraph is a standard (undirected) graph.
A \emph{path} is a series of vertives $(v_1,\dots,v_n)$ such that for each pair $(v_i,v_{i+1})$ there is
an edge $e_i \in E$ such that $\{v_i,v_{i+1}\} \subseteq e_i$.
Two vertices $a,b\in V$ are \emph{connected} if there is a path $(v_1,\dots,v_n)$ with $v_1=a$ and $v_n=b$.
The relation of being connected is an equivalence relation and induces a partition of the vertices into \emph{connected components}.

If $H$ is a $3$-hypergraph, write $Graph(H)$ for the underlying graph of $H$:
$V(Graph(H)) = V(H)$ and $(a,b) \in E(Graph(H)) \iff \exists c: (a,b,c) \in E(H)$. Note that by our definition, the connected components of a $3$-hypergraph $H$ are precisely the connected components of $Graph(H)$.

\section{Induction Lemma}

In this section, we introduce the concept of controlled gate, a generalisation of the Toffoli and Fredkin gates.
With this definition, we are able to formulate a useful induction lemma.
This lemma formalizes  the following idea.
If we can build an $(n+1)$-ary controlled gate in a certain class from
gates of arity $n$, then by replacing each $n$-ary gate with its $(n+1)$-ary extension, we have a ``spare''
control line from each $n+1$ gate, which can then be attached to an extra control input to get an $(n+2)$-ary gate.


\begin{definition}
Let $k \in \N$ and $P \subseteq B_\ell(A)$. For $w \in A^k$ and $p \in P$, define the function
$\controlled{w}{p} : A^{k+\ell} \to A^{k+\ell}$ by
\[ \controlled{w}{p}(uv) = \left\{\begin{array}{ll}
uv & \mbox{if } u \neq w \\
u p(v) & \mbox{if } u = w
\end{array}\right. \]
where $u \in A^k$, $v \in A^\ell$.
The functions $\controlled{w}{p}$, and more generally their rewirings $\pi_\alpha\circ\controlled{w}{p}\circ\pi_\alpha^{-1}$ for
$\alpha\in S_{k+\ell}$,
are called \emph{$k$-controlled $P$-permutations},
and we denote this set of functions by $CP(k,P) \subseteq B_{k+\ell}(A)$.
We refer to $CP(P) = \bigcup_k CP(k,P)$ as \emph{controlled $P$-permutations}.
\end{definition}

When $P$ is a named family of permutations, such as the family of all swaps, we usually talk about `$k$-controlled swaps' instead of `controlled swap permutations'.
The Toffoli gate is a (particular) $2$-controlled symbol permutation, while the Fredkin gate is a (particular) $1$-controlled wire swap. Note that the `$k$' in `$k$-controlled' refers to the fact that the number of controlling bits is $k$. 
Of course, sometimes we want to talk about also the particular word $w$ in $\controlled{w}{p}(uv)$. To avoid ambiguity, we say such $\controlled{w}{p}(uv)$ is \emph{$w$-word controlled permutation}. In particular, the Toffoli gate is the $11$-word controlled symbol permutation, while the Fredkin gate is a $1$-word controlled wire swap.

The following lemma formalizes the idea of adding new common control wires to all gates
in a circuit.

\begin{lemma}
\label{lem:ExtraWireLemma}
Let $k,h,\ell \in \N$, $P \subseteq B_{\ell}(A)$ and $Q \subseteq B_n(A)$. If $CP(h,Q) \subseteq \clone{CP(k,P)}$, then
$CP(h+m,Q) \subseteq\clone{CP(k+m,P)}$ for all $m \in \N$.
\end{lemma}

\begin{proof}
Consider an arbitrary
$f \in CP(h+m,Q)$. Let $uv \in A^{h+m}$ be its control word where $u\in A^m$ and $v\in A^h$, and let $p\in Q$ be its permutation.
By the hypothesis, $f_{v,p}$ 
can be implemented by maps in $CP(k,P)$.
In all their control words, add the additional input $u$. This implements $f$ as a composition of maps in $CP(k+m,P)$, as required.
\qed
\end{proof}

The main importance of the lemma comes from the following corollary:

\begin{lemma}[Induction Lemma]
\label{lem:InductionLemma}
Let  $P \subseteq B_{\ell}(A)$ be such that
$CP(k+1,P) \subseteq \clone{CP(k,P)}$ for some
$k \in \N$. Then $\clone{CP(m,P)} \subseteq \clone{CP(n,P)}$ for all $m\geq n\geq k$.
\end{lemma}

\begin{proof}
We apply Lemma~\ref{lem:ExtraWireLemma}, setting $Q=P$ and $h=k+1$.
We obtain that
$CP(k+m+1,P) \subseteq \clone{CP(k+m,P)}$ for all $m \in \N$. As $\clone{\cdot}$ is a closure operator we have that
$\clone{CP(k+m+1,P)} \subseteq \clone{CP(k+m,P)}$ for all $m \in \N$. Hence
$$
\clone{CP(k,P)} \supseteq \clone{CP(k+1,P)} \supseteq \clone{CP(k+2,P)} \supseteq \dots
$$
which clearly implies the claimed result.
\qed
\end{proof}

By the previous lemma, in order to show that a \what{} $C$ is finitely generated, it is sufficient to find
some $P \subseteq B_\ell(A)$ such that
\begin{enumerate}
\item[(i)] $\gen{CP(m,P)} = C\cap B_{m+\ell}(A)$ for all large enough $m$, and
\item[(ii)] $CP(k+1,P) \subseteq \clone{CP(k,P)}$ for some $k$.
\end{enumerate}
Indeed, if $n\geq k$ is such that (i) holds for all $m\geq n$ then,
$$
C\cap B_{m+\ell}(A) = \gen{CP(m,P)} \subseteq \clone{CP(m,P)} \subseteq \clone{CP(n,P)},
$$
where the last inclusion follows from (ii) and the Induction lemma.
Note that by (i) we also have $CP(n,P)\subseteq C$.
So the finite subset $CP(n,P)$ of $C$ generates all but finitely many elements of $C$.

Condition (i) motivates the following definition.
\begin{definition}
Let $C$ be a \what{}. We say that a set of permutations $P \subseteq B_\ell(A)$ is \emph{$n$-control-universal}
for $C$ if $\gen{CP(n-\ell, P)} = C \cap B_n(A)$. More generally, a set $P \subseteq B(A)$ that may contain gates of
different arities, is $n$-control-universal for $C$ if
$$\gen{\bigcup_{\ell}\bigcup_{f\in B_\ell(A)\cap P} CP(n-\ell, P)} = C \cap B_n(A).$$
If $P$ is $n$-control-universal
for all large enough $n$, we say it is \emph{control-universal} for $C$.
\end{definition}

In the next two sections we find gate sets that are control-universal for \what{}s of interest.


\section{Some combinatorial group theory}

In this section, we prove some basic results  that the symmetric group is generated
by any `connected' family of swaps, and the alternating group by any `connected' family of $3$-cycles.
Similar results are folklore in combinatorial group theory, but we include full proofs for completeness' sake.

Let $H$ be a graph with nodes $V(H)$ and edges $E(H)$.
The \emph{swap group} $SG(H)$ is the group $G \leq \Sym(V(H))$
generated by swaps $(a \; b)$ with $(a,b) \in E(H)$.

\begin{lemma}
\label{lem:ConCompSym}
Let $H$ be a graph with connected components $H_1, \ldots, H_k$. Then
\[ SG(H) = \Sym(V(H_1)) \times \cdots \times \Sym(V(H_k)) \]
\end{lemma}

\begin{proof}
All of the swaps act in one of the components and there are no relations between them.
Thus, the swap group will be the direct product of some permutation groups of the connected components.
We only need to show that in each connected component $H_i$, we can realize any permutation.
Since swaps generate the symmetric group, it is enough to show that if $a, b \in V(H_i)$ then
the swap $(a \; b)$ is in $SG(H)$. For this, let $a = a_0, a_1, a_2, \ldots, a_\ell = b$ be a
path from $a$ to $b$. Then
\[ (a,b) = (a_1 \; a_2) \cdots (a_{\ell-3} \; a_{\ell-2}) (a_{\ell-2} \; a_{\ell-1}) (a_{\ell} \; a_{\ell-1}) \cdots (a_3 \; a_2) (a_2 \; a_1). \]
\qed
\end{proof}


Let $H$ be a 3-hypergraph with nodes $V(H)$ and undirected edges $E(H)$. 
The \emph{cycling group} $CG(H)$ of $H$
is the group $G \leq \Sym(V(H))$ generated by cycles $(a \; b \; c)$ where $(a,b,c) \in E(H)$. 

The following observation allows us to  take any element of the alternating group given
two 3-hyperedges that intersect in one or two places.

\begin{lemma}
$$
\begin{array}{rcl}
A_4 &=& \gen{(1 \; 2 \; 3), (2 \; 3 \; 4)},\\
A_5 &=& \gen{(1 \; 2 \; 3), (3 \; 4 \; 5)}.
\end{array}
$$
\end{lemma}

\begin{lemma}
\label{lem:ConCompAlt}
Let $H$ be a hypergraph, and let the connected components of $H$ be $H_1, \ldots, H_k$. Then
\[ CG(H) = \Alt(V(H_1)) \times \Alt(V(H_2)) \times \cdots \times \Alt(V(H_k)). \]
\end{lemma}

\begin{proof}
We prove the claim by induction on the number of hyperedges.
If there are no hyperedges, then  $CG(H) = \{\ID(V(H))\}$, as required.
Now, suppose that the claim holds for a hypergraph $H'$ and $H$ is obtained from $H'$ by adding a new
hyperedge $(a, b, c)$. If none of $a, b, c$ are part of a hyperedge of $H'$ or are fully contained
in a connected component of $Graph(H')$, then the claim is trivial, as either we add a new connected
component and by definition add its alternating group $\Alt_3 \cong \gen{(a,b,c)}$ to $CG(H)$,
or we do not modify the connected components at all.

Every permutation on the right side of the equality we want to prove decomposes into even permutations in
the components. In components that do not intersect $\{a, b, c\}$, we can implement this permutation by
assumption. We thus only have to show that a pair of swaps $(x \; y) (u \; v)$ can be implemented.
If $x,y,u,v \in \{a,b,c\}$, the permutation is in $CG(H)$ by definition.
Since $(x \; y) (u \; v) = (x \; y) (a \; b)^2 (u \; v)$ it is enough to implement the permutation $(a \; b) (u \; v)$.

Now, we have two cases (up to reordering variables). Either $u \in \{a,b,c\}$ and
$v \notin \{a,b,c\}$ or $\{u,v\} \cap \{a,b,c\} = \emptyset$. 
By analysing cases, the
claim  reduces to the $\Alt_5$ or the $\Alt_4$ situation of the previous Lemma.
\qed
\end{proof}

\section{Control-universality}
\label{sec:ControlUniversality}

As corollaries of the previous section, we will now find control-universal families of gates for our  \what{}s of interest:
the full \what{} $\Cone = \bigcup_n \sym(A^n)$, the conservative \what{} $\Ctwo$, the alternating \what{}
$\Cthree=\bigcup_n \alt(A^n)$ and the alternating conservative \what{} $\Cfour$. Corollaries~\ref{cor:P1},
~\ref{cor_consuniversal}, ~\ref{cor:corollaryP3} and \ref{cor:Something} below provide control-universal gate sets
for these \what{}s.
\medskip

\noindent
{\bf a) The full \what{} $\Cone$}.
Define the graph $G^{(1)}_{A,n}$ that has nodes $A^n$ and edges $(u, v)$ where the Hamming distance between $u$ and $v$ is one.

\begin{lemma}
The graph $G^{(1)}_{A,n}$ is connected.
\end{lemma}

Let  $P_1 = \{(a \; b) \;|\; a, b \in A\} \subseteq B_1(A)$, the set of symbol swaps. The swap group of
$G^{(1)}_{A,n}$ is then $\gen{CP(n-1,P_1)}$ so,
by Lemma~\ref{lem:ConCompSym}, we have the following:

\begin{corollary}
\label{cor:P1}
For all $n$, $P_1$ is $n$-control-universal for the \what{} $\Cone$.
\end{corollary}

\medskip

\noindent
{\bf b) The conservative \what{} $\Ctwo$}.
Define the graph $G^{(2)}_{A,n}$ that has nodes $A^n$ and edges $(uabv, ubav)$ for all $a,b\in A$ and words $u,v$ with $|u|+|v|=n-2$.

\begin{lemma}
\label{lem:ConservativeConnectedComponents}
The connected components of $G^{(2)}_{A,n}$ are the weight classes.
\end{lemma}

\begin{corollary}
\label{cor_consuniversal}
Let $P_2 = \{(ab \; ba) \;|\; a,b \in A\} \subseteq B_2(A)$. Then $P_2$ is $n$-control-universal for the conservative \what{} $\Ctwo$,
for all $n\geq 1$.
\end{corollary}

The classical Fredkin gate that operates on $\{0,1\}^3$ is a $1$-controlled $P_2$-permutation.
However, note that in the case of a larger alphabet the controlled $P_2$-permutations only swap a specific pair of symbols,
not just the arbitrary contents of two cells.

We can extend this result to $Mod_k(A)$ by considering the graph as above with  added edges
 $(ua^k, ub^k)$ for all $a,b\in A$ and $u\in A^*$ with $|u|=n-k$. Then the
 set of permutations $P_2 \cup \{(a^k \; b^k) \;|\; a,b \in A\} \subseteq B_2(A) \cup B_k(A)$
 is $n$-control-universal for $Mod_k(A)$ 
 for large enough $n$.

\medskip

\noindent
{\bf c) The alternating \what{} $\Cthree$}.
Define the $3$-hypergraph $G^{(3)}_{A,n}$ that has nodes $A^n$ and hyperedges
$(uabv, uacv, udbv)$ where $a,b,c,d \in A$, $a \neq d$ and $b \neq c$, that is, all triples of words of
which two are at Hamming distance $2$ and others at distance~$1$, and the symbol
differences are in consecutive positions.

\begin{lemma}
If $n \geq 2$, then $G^{(3)}_{A,n}$ is connected. If $n = 1$, then $G^{(3)}_{A,n}$ is discrete.
\end{lemma}

\begin{corollary}
\label{cor:corollaryP3}
Let $P_3 = \{(ab \; ac \; db) \;|\; a,b,c,d \in A\} \subseteq B_2(A)$. Then $P_3$ is $n$-control-universal for the alternating \what{} $\Cthree$,
for all $n\geq 2$.
\end{corollary}

\medskip

\noindent
{\bf d) The alternating conservative \what{} $\Cfour$}.
Define the 3-hypergraph $G^{(4)}_{A,n}$ that has nodes $A^n$ and hyperedges $(uabcv, ubcav, ucabv)$ where $a,b,c$ are single symbols, that is, all (word) rotations that rotate three consecutive symbols.

\begin{lemma}
If $n > |A|$, then the connected components of $G^{(4)}_{A,n}$ are the weight classes.
\end{lemma}

\begin{proof}

When $n > |A|$ and two words $x$ and $y$ are in the same weight
class then there is an even permutation $\alpha\in S_n$ such that
$y=\pi_\alpha(x)$. This is because $x$ contains some letter twice,
say in positions $i$ and $j$, so that
$\pi_{(i\; j)}(x)=x$ for the odd permutation ${(i\; j)}\in S_n$.
The even permutation $\alpha$ is a
composition of 3-cycles of the type $(k\ k+1\ k+2)$. (To see this,
apply Lemma~\ref{lem:ConCompAlt} on the 3-hypergraph with the vertex
set $\{1,\dots ,n\}$ and hyperedges $(k,k+1,k+2)$ for $1\leq k\leq
n-2$.) But then also $\pi_\alpha$ is a composition of wire swaps of
the type $\pi_{(k\; k+1\; k+2)}$. Clearly, for all $u\in A^n$, words $u$ and $\pi_{(k\; k+1\;
k+2)}(u)$ belong to the same hyperedge of $G^{(4)}_{A,n}$ so we
conclude that $x$ and $y=\pi_\alpha(x)$ are in the same connected
component.
\qed
\end{proof}

We note that if $n \leq |A|$, then there are weight classes where each symbol occurs at most once. These classes split into two connected components depending
on the parity of the ordering of the letters.

\begin{corollary}
\label{cor:Something}
Let $P_4 = \{(abc \; bca \; cab) \;|\; a,b,c \in A\} \subseteq B_3(A)$. Then $P_4$ is $n$-control-universal
for the alternating conservative \what{} $\Cfour$, for all $n>|A|$.
\end{corollary}

\section{Finite generating sets of gates}

In order to apply the Induction Lemma we first observe that 2-controlled 3-word-cycles in any five element set can obtained from 1-controlled 3-word-cycles.

\begin{lemma}
\label{lem:fivelemma}
Let $X\subseteq A^n$ contain at least five elements, and let
\[ P=\{(x \; y \; z) \;|\; x,y,z \in X \}\subseteq B_n(A) \]
contain all $3$-word-cycles in $X$.
Then $CP(2,P) \subseteq \clone{CP(1,P)}$.
\end{lemma}

\begin{proof}
Let $x,y,z\in X$ be pairwise different, and
pick $s,t\in X$ so that $x,y,z,s,t$ are five distinct elements of $X$.
Let $p_1=(s \; t)(x \; y)$ and $p_2=(s \; t)(y \; z)$. Then $p_1$ and $p_2$ consist of two disjoint word swaps, so they are
both involutions. Moreover, $(x \; y \; z) = p_1 p_2 p_1 p_2$. Further, we have that
$$
\begin{array}{rcl}
p_1&=&(s \; t \; x)(x \; s \; y), \mbox{ and}\\
p_2&=&(s \; t \; y)(y \; s \; z).
\end{array}
$$
Let $a,b\in A$ be arbitrary and consider the $2$-controlled $P$-permutation $f=\controlled{ab}{(x \; y \; z)} \in B_{2+n}(A)$ determined by the control word $ab$
and the 3-word-cycle $(x \; y \; z)$. Then $f=g \circ g$ where
$$
g = \controlled{a*}{p_1}\circ \controlled{*b}{p_2} = \controlled{a*}{(s \; t \; x)} \circ \controlled{a*}{(x \; s \; y)} \circ \controlled{*b}{(s \; t \; y)} \circ \controlled{*b}{(y \; s \; z)}
$$
is a composition of four 1-controlled $P$-permutations, where the star symbol indicates the control symbol not used by the gate.
See Figure~\ref{fig:threecycles} for an illustration.

\begin{figure}[ht]
\begin{center}
\includegraphics[width=11.5cm]{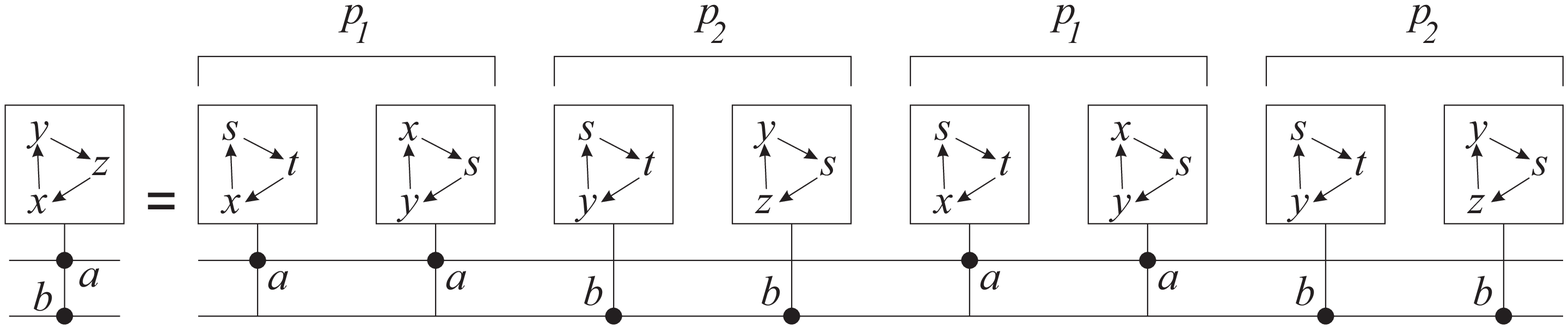}
\end{center}
\caption{A decomposition of the $ab$-controlled 3-word-cycle $(x \; y \; z)$ into a composition of eight 1-controlled 3-word-cycles.}
\label{fig:threecycles}
\end{figure}

To verify that indeed $f=g \circ g$, consider an input $w=a'b'u$ where $a',b'\in A$ and $u\in A^n$. If $a'\neq a$ then
$g(w)=\controlled{*b}{p_2}(w)$, so that $g\circ g(w)=w=f(w)$ since $p_2$ is an involution. Analogously, if $b'\neq b$ then $g\circ g(w)=w=f(w)$, because $p_1$ is an involution.
Suppose then that $a'=a$ and $b'=b$. We have $g\circ g(w)=ab( (p_1 p_2 p_1 p_2)(u)) = f(w)$. We conclude that $f\in \clone{CP(1,P)}$, and because
$f$ was an arbitrary element of $CP(2,P)$, up to reordering the input and output symbols, the claim $CP(2,P)\subseteq \clone{CP(1,P)}$ follows.
\qed
\end{proof}

\begin{corollary}
\label{cor:threecycles}
Let $X\subseteq A^n, P\subseteq B_n(A)$ be as in Lemma~\ref{lem:fivelemma}. Then
$\clone{CP(m,P)} \subseteq \clone{CP(1,P)}$
for all $m\geq 1$.

\end{corollary}

\begin{proof}
Apply Lemma~\ref{lem:InductionLemma} with $k=1$.
\qed
\end{proof}

\subsection{The alternating and full \what{}s}

Assuming that $|A| > 1$, the set $X=A^3$ contains at least five elements.
For $P=\{(x \; y \; z) \;|\; x,y,z \in A^3\}\subseteq B_3(A)$ we then have, by Corollary~\ref{cor:threecycles}, that
$\clone{CP(m,P)} \subseteq \clone{CP(1,P)}$ for all $m\geq 1$.

Recall that $P_3=\{(ab \; ac \; db) \;|\; a,b,c,d \in A\} \subseteq B_2(A)$ is $n$-control-universal for the alternating \what{} $\Cthree$, for $n\geq 2$ (Corollary~\ref{cor:corollaryP3}).
Clearly $\cp{1}{P_3}\subseteq P \subseteq \clone{\cp{0}{P}}$, so by Lemma~\ref{lem:ExtraWireLemma}, for any $m\geq 1$,
$$
\cp{m+1}{P_3} \subseteq \clone{\cp{m}{P}} \subseteq \clone{\cp{1}{P}}.
$$
Hence $\Cthree\cap B_{m+3}(A) = \gen{CP(m+1,P_3)} \subseteq \clone{\cp{1}{P}}$.
We conclude that $\clone{\cp{1}{P}}$ contains all permutations of $\Cthree$ except the ones in $B_1(A), B_2(A)$ and $B_3(A)$. We have proved the following theorem.

\begin{theorem}
\label{thm:AltFiniteGen}
The alternating \what{} $\Cthree$ is finitely generated. Even permutations of $A^4$ generate all even permutations of $A^n$ for all $n\geq 4$.
\end{theorem}

\begin{corollary}
\label{cor:FullFiniteGen}
Let $|A|$ be odd. Then the full \what{} $\Cone$ is finitely generated.
The permutations of $A^4$ generate all permutations of $A^n$ for all $n\geq 4$.
\end{corollary}

\begin{proof}
Let $|A|>1$ be odd.
Let $P$ be the set of all permutations of $A^4$, and let $n\geq 4$. By Theorem~\ref{thm:AltFiniteGen},
the closure $\clone{P}$ contains all even permutations of $A^n$. The set
$P$ also contains an odd permutation $f$, say the word swap $(0000\; 1000)$.
Consider
$\pi=f\oplus \id_{n-4}\in B_n(A)$  that applies the swap $f$ on the first
four input symbols and keeps the others unchanged.
This $\pi$ is an odd permutation because it consists of $|A|^{m-4}$ disjoint swaps and $|A|$ is odd.
Because $\clone{P}\cap B_n(A)$ contains all even permutations of $A^n$ and an odd one, it contains all permutations.
\qed
\end{proof}

Recall that if
a circuit implements the permutation $f\oplus\id_k\in B_{n+k}(A)$,
we say it implements $f \in B_n(A)$ using $k$
borrowed bits.

\begin{corollary}
\label{cor:borrowedcor}
The \what{} $\Cone$ is finitely generated using at most one borrowed bit.
\end{corollary}

\begin{proof}
For $|A|$ odd the claim follows from Corollary~\ref{cor:FullFiniteGen}.
When $A$ is even then the permutations $f\oplus\id$ with one borrowed bit are all even, so the claim follows
from Theorem~\ref{thm:AltFiniteGen}.
\qed
\end{proof}

\subsection{The alternating conservative \what{}}

Assuming $|A|>1$, every non-trivial weight class of $A^5$ contains at least five elements.
(The trivial weight-classes are the singletons $\{a^5\}$ for $a\in A$.) For every non-trivial weight class $X$ we set
$P_X=\{(x \; y \; z) \;|\; x,y,z \in X\}\subseteq B_5(A)$ for the 3-word-cycles in $X$. By Corollary
\ref{cor:threecycles} we know that
$\clone{CP(m,P_X)} \subseteq \clone{CP(1,P_X)}$ for all $m\geq 1$. Let $P$ be the union of $P_X$ over all non-trivial weight classes $X$.
Then, because $\clone{\cdot}$ is a closure operator,  also $\clone{CP(m,P)} \subseteq \clone{CP(1,P)}$ for all $m\geq 1$.

By Corollary~\ref{cor:Something}, the set $P_4 = \{(abc \; bca \; cab) \;|\; a,b,c \in A\} \subseteq B_3(A)$
is $n$-control-universal for the alternating conservative
\what{} $\Cfour$, for all $n > |A|$.

Let $m\in\N$ be such that $m\geq 1$ and $m+5 > |A|$.
Because $\cp{2}{P_4}\subseteq P \subseteq \clone{\cp{0}{P}}$, by Lemma~\ref{lem:ExtraWireLemma} we have
$$
\cp{m+2}{P_4} \subseteq \clone{\cp{m}{P}} \subseteq \clone{\cp{1}{P}}.
$$
Hence $\Cfour\cap B_{m+5}(A) = \gen{CP(m+2,P_4)} \subseteq \clone{\cp{1}{P}}$.
We conclude that $\clone{\cp{1}{P}}$ contains all permutations of $\Cfour$
except possibly the ones in
$B_k(A)$ for $k\leq 5$ and for $k\leq |A|$.
This proves the following theorem.

\begin{theorem}
The alternating conservative \what{} $\Cfour$ is finitely generated. A gate set generates the whole $\Cfour$ if it generates,
for all $n\leq 6$ and all $n\leq |A|$,
the conservative permutations of $A^n$ that are even on all weight classes.
\qed
\end{theorem}

\section{Non-finitely generated \what{}s}

It is well known that the full \what{} is not finitely generated over even alphabets. The reason is that any permutation
$f\in B_n(A)$ can only compute even permutations on $A^{m}$ for $m>n$.
\begin{theorem}[\cite{toff80}]
For even $|A|$, the full \what{} $\Cone$ is not finitely generated.
\end{theorem}

By another parity argument we can also show that the conservative \what{} $\Ctwo$ is not finitely generated on any
non-trivial alphabet, not even if infinitely many borrowed bits are available.
This generalizes a result in~\cite{xu15} on binary alphabets. Our proof is based on the same parity sequences as the one
in~\cite{xu15}, where these sequences are computed concretely for generalized Fredkin gates. However, our observation only relies on the (necessarily) low rank of a finitely-generated group of such parity sequences, and the particular conserved quantity is not as important.


Let $n\in\N$, and let $W$ be the family of the weight classes of $A^n$.
For any  $f\in \Cthree\cap B_n(A)$ and any weight class $c\in W$,
the restriction $f|_{c}$ of $f$ on the weight class $c$ is a permutation of $c$.
Let $\phi(f)_c\in \Z_2$ be its parity. Clearly, $\phi(f\circ g)_c=\phi(f)_c+\phi(g)_c$ modulo two,
so $\phi$ defines a group homomorphism from $\Cthree\cap B_n(A)$ to the additive abelian group $(\Z_2)^W$.
The image $\phi(f)$ that records all $\phi(f)_c$ for all $c\in W$ is the \emph{parity sequence} of $f$.
Because each element of the commutative group $(\Z_2)^W$ is an involution, it follows that the subgroup generated by
any $k$ elements has cardinality at most $2^k$.

Consider then a function $f\in \Cthree\cap B_{\ell}(A)$ for $\ell\leq n$.
Its application $f_n = f\oplus\id_{n-\ell} \in B_n(A)$
on length $n$ inputs is conservative, so it has the associated parity sequence $\phi(f')$, which we
denote by $\phi_n(f)$. Note that any conjugate $gfg^{-1}$ of $f$ by a wire permutation $g$ has the same parity sequence, so the parity sequence
does not depend on which input wires we apply $f$ on.

Let $f^{(1)}, f^{(2)}, \dots , f^{(m)}\in \Ctwo$  be a finite generator set, and
let us denote by $C\subseteq \Ctwo$ the \what{} they generate.
Let $n \geq 2$ be larger than the arity of any $f^{(i)}$. Then $C \cap B_n(A)$
is the group generated by the applications $f^{(1)}_n, f^{(2)}_n, \dots , f^{(m)}_n$
of the generators on length $n$ inputs, up to conjugation by wire permutations.
We conclude that there are at most $2^m$ different parity sequences on $C\cap B_n(A)$, for all sufficiently large $n$. We have proved the following lemma.


\begin{lemma}
\label{lem:paritysequencelemma}
Let $C$ be a finitely generated sub\what{} of $\Ctwo$. Then there exists a constant $N$ such that, for all $n$,
the elements of $C\cap B_n(A)$ have at most $N$ different parity sequences.
\end{lemma}

Now we can prove the following negative result. Not only does it state
that no finite gate set generates the conservative \what{}, but even that there
necessarily remain conservative permutations that cannot be obtained
using any number of borrowed bits.

\begin{theorem}
\label{thm:XuGeneralization}
Let $|A|>1$. The conservative \what{} $Cons(A)$ is not finitely generated. In fact, if
$C\subseteq Cons(A)$ is finitely generated then there exists $f\in
Cons(A)$ such that $f\oplus \id_k\not\in C$ for all $k=0,1,2,\dots$.
\end{theorem}

\begin{proof}
Let $0,1 \in A$ be distinct.
Let $C$ be a finitely generated sub\what{} of $\Ctwo$, and
let $N$ be the constant from
Lemma~\ref{lem:paritysequencelemma} for $C$. Let us fix $n\geq
N+2$. For each $i=1,2,\dots, N+1$,
consider the non-trivial weight classes $c_i$ containing the words of $A^n$ with $i$ letters 1 and $n-i$ letters 0.
For each $i$, let $f_i$ be the the permutation $f_i\in \Ctwo\cap B_n(A)$ that swaps two elements of $c_i$,
keeping all other elements of $A^n$ unchanged. This $f_i$ is odd on $c_i$ and even on all other weight classes, so all $f_i$ have different parity sequences.
We conclude that
some $f_i$ is not in $C$.

For the second, stronger claim, we continue by considering
an arbitrary $k\in\N$. For $i=1,2,\dots , N+1$, let $c^{(k)}_i$ be the parity class
of $A^{n+k}$ containing the words with $i$ letters 1 and $n+k-i$ letters 0.
Note that  $f^{(k)}_i=f_i\oplus \id_k$ is odd on $c^{(k)}_i$ and even on all $c^{(k)}_j$ with $j<i$. This means that
the parity sequences of $f^{(k)}_1, f^{(k)}_2, \dots , f^{(k)}_{N+1}$ are all different, hence some $f^{(k)}_i$ is not in $C$.
But then, for some $i\in\{1,2,\dots, N+1\}$, there are infinitely many
$k\in\N$ with the property that $f^{(k)}_i=f_i\oplus \id_k$ is not in
$C$. This means that $f_i\oplus \id_k\not\in C$ for \emph{any} $k\in \N$
as $f_i\oplus \id_k\in C$ implies that $f_i\oplus \id_\ell\in C$ for
all $\ell>k$. The permutation $f=f_i$ has the claimed property.
\qed
\end{proof}

The theorem generalizes directly to \what{}s defined by a certain type of conserved quantities, at least when borrowed bits are not used.

\begin{definition}
Let $|A| > 1$ and let $\sim$ be a sequence of equivalence relations, so that for all $n$, $\sim_n$ is an equivalence relation on $A^n$. If
\[ u \sim_n v \implies ua \sim_{n+1} va \]
then we say $\sim$ is \emph{compatible}, and if
\[ u \sim_n v \implies \pi(u) \sim_n \pi(v) \]
for all wire permutations $\pi$, then we say $\sim$ is \emph{permutable}. We say $\sim$ is a \emph{generalized conserved quantity} if it is both compatible and permutable. If for all $m \in \N$, there exists $n$ such that $\sim_n$ has at least $m$ equivalence classes with more than one word, we say $\sim$ is \emph{infinite-dimensional}.
\end{definition}

Say that $f \in B_n(A)$ is $\sim$-preserving if $f(u) \sim_{|u|} u$ for all $u \in \bigcup_n A^n$, and write $C_{\sim}$ for the set of all $\sim$-preserving permutations. 

\begin{theorem}
\label{thm:ConservedQuantityNonFG}
If $\sim$ is a generalized conserved quantity, then $C_{\sim}$ is a \what{}. If $\sim$ is infinite-dimensional, then $C_{\sim}$ is not finitely generated.
\end{theorem}

The theorem shows, for example, that the \what{} of functions in $B(\{0,1,2\})$ that preserve the number of zeroes, and preserve the number of ones modulo $k$, is not finitely generated.

\section{Concrete generating families}
\label{secsearches}

We have found finite generating sets for \what{}s in both the general and the conservative case. Our generating sets are of the form `all controlled $3$-word cycles that are in the family', and the reader may wonder whether there are more natural gate families that generate these classes. Of course, by our results, there is an algorithm for checking whether a particular set of gates is a set of generators, and in this section we give some examples.

First, we observe that $CP(2,P_1)$ (that is, $2$-controlled symbol swaps) generate all permutations of $A^3$ and all even permutations of $A^n$ for all $n\geq 4$. Indeed, by Corollary~\ref{cor:P1} they generate $B_3(A)$, and by Figure~\ref{fig:tof} they generate $CP(2,P_3)$ (the 2-controlled $3$-cycles of length-two words). These in turn, by Corollary~\ref{cor:corollaryP3}, generate all even permutations of $A^4$
which is enough by Theorem~\ref{thm:AltFiniteGen} to get all even permutations on $A^n$ for $n\geq 4$.

It is easy to see that $CP(2,P_1)$ in turn is generated by all symbol swaps and the $w$-word-controlled symbol swaps for a single $w \in A^2$. In particular in the case of binary alphabets, we obtain that the alternating \what{} is generated by the Toffoli gate and the negation gate, which was also proved in \cite{xu15}.

In the conservative binary case, the Fredkin gate is known to be universal (in the sense of auxiliary bits, see \cite{xu15}). The Fredkin gate is, due to the binary alphabet, both the unique 1-word-controlled wire swap and the unique nontrivial conservative 1-word-controlled word swap. The natural generalizations would be to show that in general the $1$-controlled wire swaps or conservative word swaps generate the alternating conservative \what{}. We do not prove this, but do show how the universality of the Fredkin gate follows from our results and a bit of computer search.

\begin{figure}[ht]
\begin{center}
\includegraphics[width=9cm]{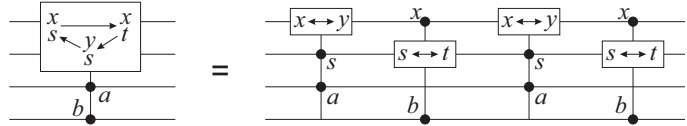}
\end{center}
\caption{A decomposition of the $ab$-controlled 3-cycle $(xs \; xt \; ys)$ into a composition of four 2-controlled swaps.}
\label{fig:tof}
\end{figure}

The following shows that the 00-word-controlled rotation is generated by the 0-word-controlled rotation.

\begin{lemma}
The $00$-word-controlled three-wire rotation can be implemented with nine $0$-word-controlled three-wire rotations but can not be implemented with eight. The $01$-word-controlled three-wire rotation can be implemented with eight $0$-word-controlled three-wire rotations but can not be implemented with seven.
\end{lemma}

\begin{proof}
A computer search shows that eight and seven gates do not suffice. We show how to compose the $00$-word-controlled rotation out of nine $0$-word-controlled rotations.

Let $A = \{0,1\}$ and $R \in B_3(A)$ be the rotation $R = \pi_{(1\,2\,3)}$. Write $\rho_{a,b,c,d}(f)$ for $f$ applied to cells $a,b,c,d$ in that order.
\begin{align*}
\controlled{00}{R} = \;
&\rho_{1,0,2,3}(\controlled{0}{R}) \circ
\rho_{3,1,4,2}(\controlled{0}{R}) \circ
\rho_{1,0,2,4}(\controlled{0}{R}) \circ \\
&\rho_{3,0,1,2}(\controlled{0}{R}) \circ
\rho_{0,1,3,4}(\controlled{0}{R}) \circ
\rho_{1,2,3,4}(\controlled{0}{R}) \circ \\
&\rho_{0,1,4,3}(\controlled{0}{R}) \circ
\rho_{1,0,2,3}(\controlled{0}{R}) \circ
\rho_{3,0,2,4}(\controlled{0}{R})
\end{align*}

See Figure~\ref{fig:ccrot} for the diagrams of both this, and the implementation of the $01$-word-controlled three-word rotation.
\qed
\end{proof}

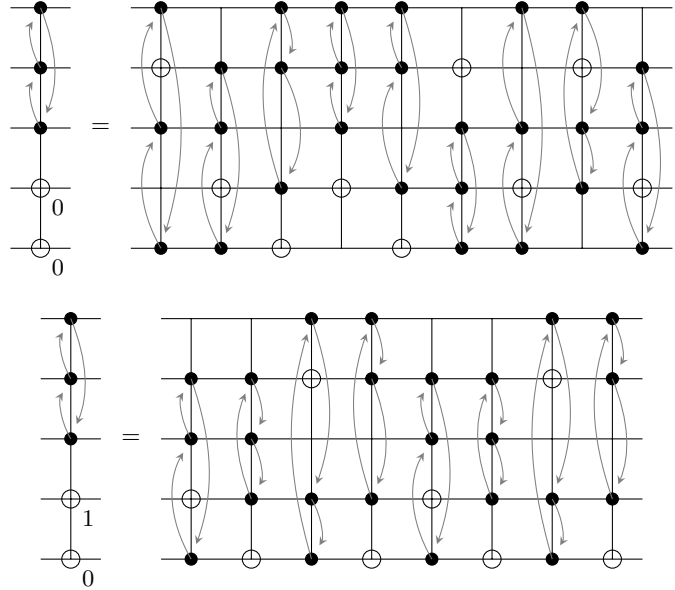
\begin{figure}[h]
\begin{center}
 \begin{tikzpicture}[scale = 0.8]

\draw (-2.5,0) grid (-1.5,4);
\draw (-2,0) circle (0.15);
\draw (-2,1) circle (0.15);
\draw[fill] (-2,2) circle (0.1);
\draw[fill] (-2,3) circle (0.1);
\draw[fill] (-2,4) circle (0.1);
\draw[-stealth,shorten >=6pt,color=gray] (-2,2) to[bend left=30] (-2,3);
\draw[-stealth,shorten >=6pt,color=gray] (-2,3) to[bend left=30] (-2,4);
\draw[-stealth,shorten >=6pt,color=gray] (-2,4) to[bend left=22] (-2,2);
\node[below right=0.05] () at (-2,0) {$0$}; \node[below right=0.05] () at (-2,1) {$0$};

\node () at (-1,2) {$=$};

\draw (-0.5,0) grid (8.5,4);

\draw (0,3) circle (0.15); \draw (1,1) circle (0.15); \draw (2,0) circle (0.15);
\draw (3,1) circle (0.15); \draw (4,0) circle (0.15); \draw (5,3) circle (0.15);
\draw (6,1) circle (0.15); \draw (7,3) circle (0.15); \draw (8,1) circle (0.15);

\draw[fill] (0,0) circle (0.1); \draw[fill] (1,0) circle (0.1); \draw[fill] (2,1) circle (0.1);
\draw[fill] (3,2) circle (0.1); \draw[fill] (4,1) circle (0.1); \draw[fill] (5,0) circle (0.1);
\draw[fill] (6,0) circle (0.1); \draw[fill] (7,1) circle (0.1); \draw[fill] (8,0) circle (0.1);

\draw[fill] (0,2) circle (0.1); \draw[fill] (1,2) circle (0.1); \draw[fill] (2,3) circle (0.1);
\draw[fill] (3,3) circle (0.1); \draw[fill] (4,3) circle (0.1); \draw[fill] (5,1) circle (0.1);
\draw[fill] (6,2) circle (0.1); \draw[fill] (7,2) circle (0.1); \draw[fill] (8,2) circle (0.1);

\draw[fill] (0,4) circle (0.1); \draw[fill] (1,3) circle (0.1); \draw[fill] (2,4) circle (0.1);
\draw[fill] (3,4) circle (0.1); \draw[fill] (4,4) circle (0.1); \draw[fill] (5,2) circle (0.1);
\draw[fill] (6,4) circle (0.1); \draw[fill] (7,4) circle (0.1); \draw[fill] (8,3) circle (0.1);


\draw[-stealth,shorten >=6pt,color=gray] (0,0) to[bend left=30] (0,2);
\draw[-stealth,shorten >=6pt,color=gray] (0,2) to[bend left=30] (0,4);
\draw[-stealth,shorten >=6pt,color=gray] (0,4) to[bend left=16] (0,0);

\draw[-stealth,shorten >=6pt,color=gray] (1,0) to[bend left=30] (1,2);
\draw[-stealth,shorten >=6pt,color=gray] (1,2) to[bend left=30] (1,3);
\draw[-stealth,shorten >=6pt,color=gray] (1,3) to[bend left=19] (1,0);

\draw[-stealth,shorten >=6pt,color=gray] (2,4) to[bend left=30] (2,3);
\draw[-stealth,shorten >=6pt,color=gray] (2,3) to[bend left=30] (2,1);
\draw[-stealth,shorten >=6pt,color=gray] (2,1) to[bend left=19] (2,4);

\draw[-stealth,shorten >=6pt,color=gray] (3,2) to[bend left=30] (3,3);
\draw[-stealth,shorten >=6pt,color=gray] (3,3) to[bend left=30] (3,4);
\draw[-stealth,shorten >=6pt,color=gray] (3,4) to[bend left=22] (3,2);

\draw[-stealth,shorten >=6pt,color=gray] (4,1) to[bend left=30] (4,3);
\draw[-stealth,shorten >=6pt,color=gray] (4,3) to[bend left=30] (4,4);
\draw[-stealth,shorten >=6pt,color=gray] (4,4) to[bend left=19] (4,1);

\draw[-stealth,shorten >=6pt,color=gray] (5,0) to[bend left=30] (5,1);
\draw[-stealth,shorten >=6pt,color=gray] (5,1) to[bend left=30] (5,2);
\draw[-stealth,shorten >=6pt,color=gray] (5,2) to[bend left=22] (5,0);

\draw[-stealth,shorten >=6pt,color=gray] (6,0) to[bend left=30] (6,2);
\draw[-stealth,shorten >=6pt,color=gray] (6,2) to[bend left=30] (6,4);
\draw[-stealth,shorten >=6pt,color=gray] (6,4) to[bend left=16] (6,0);

\draw[-stealth,shorten >=6pt,color=gray] (7,4) to[bend left=30] (7,2);
\draw[-stealth,shorten >=6pt,color=gray] (7,2) to[bend left=30] (7,1);
\draw[-stealth,shorten >=6pt,color=gray] (7,1) to[bend left=19] (7,4);

\draw[-stealth,shorten >=6pt,color=gray] (8,0) to[bend left=30] (8,2);
\draw[-stealth,shorten >=6pt,color=gray] (8,2) to[bend left=30] (8,3);
\draw[-stealth,shorten >=6pt,color=gray] (8,3) to[bend left=19] (8,0);
\end{tikzpicture}
\end{center}

\begin{center}
\begin{tikzpicture}[scale = 0.8]

\draw (-2.5,0) grid (-1.5,4);
\draw (-2,0) circle (0.15);
\draw (-2,1) circle (0.15);
\draw[fill] (-2,2) circle (0.1);
\draw[fill] (-2,3) circle (0.1);
\draw[fill] (-2,4) circle (0.1);
\draw[-stealth,shorten >=6pt,color=gray] (-2,2) to[bend left=30] (-2,3);
\draw[-stealth,shorten >=6pt,color=gray] (-2,3) to[bend left=30] (-2,4);
\draw[-stealth,shorten >=6pt,color=gray] (-2,4) to[bend left=22] (-2,2);
\node[below right=0.05] () at (-2,0) {$0$}; \node[below right=0.05] () at (-2,1) {$1$};

\node () at (-1,2) {$=$};

\draw (-0.5,0) grid (7.5,4);

\draw (0,1) circle (0.15); \draw (1,0) circle (0.15); \draw (2,3) circle (0.15);
\draw (3,0) circle (0.15); \draw (4,1) circle (0.15); \draw (5,0) circle (0.15);
\draw (6,3) circle (0.15); \draw (7,0) circle (0.15);

\draw[fill] (0,0) circle (0.1); \draw[fill] (1,1) circle (0.1); \draw[fill] (2,0) circle (0.1);
\draw[fill] (3,1) circle (0.1); \draw[fill] (4,0) circle (0.1); \draw[fill] (5,1) circle (0.1);
\draw[fill] (6,0) circle (0.1); \draw[fill] (7,1) circle (0.1);

\draw[fill] (0,2) circle (0.1); \draw[fill] (1,3) circle (0.1); \draw[fill] (2,4) circle (0.1);
\draw[fill] (3,4) circle (0.1); \draw[fill] (4,2) circle (0.1); \draw[fill] (5,3) circle (0.1);
\draw[fill] (6,4) circle (0.1); \draw[fill] (7,4) circle (0.1);

\draw[fill] (0,3) circle (0.1); \draw[fill] (1,2) circle (0.1); \draw[fill] (2,1) circle (0.1);
\draw[fill] (3,3) circle (0.1); \draw[fill] (4,3) circle (0.1); \draw[fill] (5,2) circle (0.1);
\draw[fill] (6,1) circle (0.1); \draw[fill] (7,3) circle (0.1);


\draw[-stealth,shorten >=6pt,color=gray] (0,0) to[bend left=30] (0,2);
\draw[-stealth,shorten >=6pt,color=gray] (0,2) to[bend left=30] (0,3);
\draw[-stealth,shorten >=6pt,color=gray] (0,3) to[bend left=19] (0,0);

\draw[-stealth,shorten >=6pt,color=gray] (1,1) to[bend left=22] (1,3);
\draw[-stealth,shorten >=6pt,color=gray] (1,3) to[bend left=30] (1,2);
\draw[-stealth,shorten >=6pt,color=gray] (1,2) to[bend left=30] (1,1);

\draw[-stealth,shorten >=6pt,color=gray] (2,0) to[bend left=16] (2,4);
\draw[-stealth,shorten >=6pt,color=gray] (2,4) to[bend left=19] (2,1);
\draw[-stealth,shorten >=6pt,color=gray] (2,1) to[bend left=30] (2,0);

\draw[-stealth,shorten >=6pt,color=gray] (3,1) to[bend left=19] (3,4);
\draw[-stealth,shorten >=6pt,color=gray] (3,4) to[bend left=30] (3,3);
\draw[-stealth,shorten >=6pt,color=gray] (3,3) to[bend left=22] (3,1);

\draw[-stealth,shorten >=6pt,color=gray] (4,0) to[bend left=30] (4,2);
\draw[-stealth,shorten >=6pt,color=gray] (4,2) to[bend left=30] (4,3);
\draw[-stealth,shorten >=6pt,color=gray] (4,3) to[bend left=19] (4,0);

\draw[-stealth,shorten >=6pt,color=gray] (5,1) to[bend left=22] (5,3);
\draw[-stealth,shorten >=6pt,color=gray] (5,3) to[bend left=30] (5,2);
\draw[-stealth,shorten >=6pt,color=gray] (5,2) to[bend left=30] (5,1);

\draw[-stealth,shorten >=6pt,color=gray] (6,0) to[bend left=16] (6,4);
\draw[-stealth,shorten >=6pt,color=gray] (6,4) to[bend left=19] (6,1);
\draw[-stealth,shorten >=6pt,color=gray] (6,1) to[bend left=30] (6,0);

\draw[-stealth,shorten >=6pt,color=gray] (7,1) to[bend left=19] (7,4);
\draw[-stealth,shorten >=6pt,color=gray] (7,4) to[bend left=30] (7,3);
\draw[-stealth,shorten >=6pt,color=gray] (7,3) to[bend left=22] (7,1);

\end{tikzpicture}
\end{center}

\caption{Diagrams for a $00$-controlled rotation and a $01$-controlled rotation built from $0$-controlled rotations. The rotations are controlled by the two bottommost wires, and the rotation rotates the wires in order $2 \rightarrow 3 \rightarrow 4 \rightarrow 2$, where the bottommost wire is the $0$th one. The diagram is read from left to right, and on each column we perform a $0$-controlled rotation. The large circle indicates the control wire, and the dots are the rotated wires. The arrows indicate the direction of rotation. 
}
\label{fig:ccrot}
\end{figure}

\begin{lemma}
The word cycle $(0001 \; 0010 \; 0100)$ can be built from six $0$-word-controlled three-wire rotations (but no less). The same is true for $(0011 \; 0110 \; 0101)$.
\end{lemma}

\begin{proof}
This can be proved by a short brute force search. \qed
\end{proof}


Let $\pi_1 = (001 \; 010 \; 100)$ and $\pi_2 = (011 \; 110 \; 101)$. Note that $\pi_1 \circ \pi_2$ is the three-wire rotation. Then, by the first lemma of this section and Lemma~\ref{lem:InductionLemma}, we have that $1$-control $(\pi_1 \circ \pi_2)$-permutations generate $k$-controlled $(\pi_1 \circ \pi_2)$-permutations for all $k$. By the second lemma of this section, $1$-controlled $(\pi_1 \circ \pi_2)$-permutations generate $1$-controlled $\{\pi_1, \pi_2\}$-permutations, so by Lemma~\ref{lem:ExtraWireLemma}, $k$-controlled $(\pi_1 \circ \pi_2)$-permutations generate $k$-controlled $\{\pi_1, \pi_2\}$-permutations for all $k$. Putting these together and combining with Corollary~\ref{cor:Something}, we have:

\begin{theorem}
Let $A = \{0,1\}$. Then the alternating conservative \what{} $\Cfour$ is generated by the controlled wire rotation
\[ f(a,b,c,d) = \left\{\begin{array}{cc}
(a,c,d,b) & \mbox{if } a = 0 \\
(a,b,c,d) & \mbox{otherwise}
\end{array}\right. \]
and the even conservative permutations of $A^3$.
\end{theorem}

Clearly $f(a,b,c,d)$ is generated by $1$-controlled wire swaps. It follows that the Fredkin gate together with the (unconditional) wire swap generates all even conservative permutations of $\{0,1\}^n$ for $n \geq 4$.





\section{Conclusion}


We have been able to precisely determine the \what{} generated by a finite set of generators over an even order alphabet and show that over an odd alphabet, a finite collection of mappings generates the whole \what{}. The first result confirms a conjecture in \cite{boykett15} and the second gives a simpler proof of the same result from that paper. Moreover, we have shown that the alternating conservative \what{} is finitely generated on all alphabets, but the conservative \what{} is never finitely generated.

The methods are rather general: We have developed an induction result (Lemma~\ref{lem:InductionLemma}) for finding generating sets for \what{}s of controlled permutations, allowing us to determine finite generating sets for some \what{}s with uniform methods. We also prove the nonexistence of a finite generating family for conserved gates with a general method in Theorem~\ref{thm:ConservedQuantityNonFG}, when borrowed bits are not used. We only need particular properties of the weight function in the proof of Theorem~\ref{thm:XuGeneralization}, where it is shown that the (usual) conservative \what{} is not finitely generated even when borrowed bits are allowed.

In \cite{aaronsonetal15} the full list of reversible gate families in the binary case is listed, when the use of auxiliary bits is allowed. This includes the conservative \what{}, various modular \what{}s and nonaffine \what{}s. As we do not allow the use of auxiliary bits, we are not limited to these \what{}s; still, it is an interesting question which of them are finitely generated in our strict sense. 


While this paper develops strong techniques for showing finitely generatedness and non-finitely generatedness of \what{}s, our generating sets are rather abstract, and do not correspond very well to known generating sets. It would be of value to replace the constructions found by computer search in section \ref{secsearches} by more understandable constructions, in order to find more concrete generating sets in the case of general alphabets in the case of conservative gates. 

\bibliographystyle{splncs03}
\bibliography{revcompproc}

\def\cprime{$'$}
\begin{thebibliography}{1}
\providecommand{\url}[1]{\texttt{#1}}
\providecommand{\urlprefix}{URL }

\bibitem{aaronsonetal15}
Aaronson, S., Grier, D., Schaeffer, L.: The classification of reversible bit
  operations. Electronic Colloquium on Computational Complexity (66) (2015)

\bibitem{boykett15}
Boykett, T.: Closed systems of invertible maps (2015),
  \url{http://arxiv.org/abs/1512.06813}, submitted

\bibitem{FrTo82}
Fredkin, E., Toffoli, T.: Conservative logic. International Journal of
  Theoretical Physics  21(3),  219--253 (1982),
  \url{http://dx.doi.org/10.1007/BF01857727}

\bibitem{szendrei}
Szendrei, {\'A}.: Clones in universal algebra, S\'eminaire de Math\'ematiques
  Sup\'erieures [Seminar on Higher Mathematics], vol.~99. Presses de
  l'Universit\'e de Montr\'eal, Montreal, QC (1986)

\bibitem{toff80}
Toffoli, T.: Reversible computing. Tech. Rep. MIT/LCS/TM-151, MIT (1980)

\bibitem{xu15}
Xu, S.: Reversible Logic Synthesis with Minimal Usage of Ancilla Bits. Master's
  thesis, MIT (June 2015), \url{http://arxiv.org/pdf/1506.03777.pdf}

\bibitem{yangetal05}
Yang, G., Song, X., Perkowski, M., Wu, J.: Realizing ternary quantum switching
  networks without ancilla bits. J. Phys. A  38(44),  9689--9697 (2005),
  \url{http://dx.doi.org/10.1088/0305-4470/38/44/006}

\end{thebibliography}

\end{document}